\documentclass[12pt]{amsart}
\usepackage{amssymb,latexsym}
\usepackage{amsfonts}
\usepackage{amsmath}
\usepackage[colorlinks,linkcolor=blue,anchorcolor=blue,citecolor=blue]{hyperref}
\usepackage{algorithm}
\usepackage{enumerate}
\usepackage{algpseudocode}
\usepackage{verbatim}
\usepackage{graphicx}

\newcommand{\tr}{{\rm Tr}}

\newcommand{\F}{{\mathbb F}}

\newcommand{\cC}{{\mathcal C}}
\newcommand{\gS}{{\mathrm S}}
\newcommand{\B}{{\mathbb B}}

\newtheorem{theorem}{Theorem}[section]
\newtheorem{lemma}[theorem]{Lemma}

\theoremstyle{definition}
\newtheorem{definition}[theorem]{Definition}

\newtheorem{corollary}[theorem]{Corollary}

\theoremstyle{remark}
\newtheorem{remark}[theorem]{Remark}

\numberwithin{equation}{section}

\begin{document}

\title{Optimal Combinatorial Neural Codes with Matched Metric $\delta_{r}$: Characterization and Constructions}

\author{Aixian Zhang}
\address{Department of Mathematical Sciences, Xi'an  University of Technology,
Shanxi, 710054, China.}
\email{zhangaixian1008@126.com}

\author{Xiaoyan. Jin}
\address{  Department of Mathematical
Sciences, North West University,
 Shannxi, 710127, China.}
\email{jxymg@126.com }

\author{Keqin Feng}
\address{Department of Mathematical
Sciences, Tsinghua University,
 Beijing, 100084, China. }
\email{fengkq@tsinghua.edu.cn}

\subjclass[2010]{11T06, 11T55}



\keywords{Combinatorial neural codes, asymmetric discrepancy, bent functions,
Hamming, Singleton and Plotkin bounds..}

\begin{abstract}
Based on the theoretical neuroscience, G. Cotardo and A. Ravagnavi in  \cite{CR} introduced a kind of
asymmetric binary codes called combinatorial neural codes (CN codes for short),
with a ``matched metric" $\delta_{r}$ called asymmetric discrepancy, instead of the Hamming distance
$d_{H}$ for usual error-correcting codes. They also presented the Hamming, Singleton and Plotkin bounds for CN codes
with respect to $\delta_{r}$ and asked how to construct the CN codes $\cC$ with large size $|\cC|$ and $\delta_{r}(\cC).$
In this paper we firstly show that a binary code $\cC$ reaches one of the above bounds for $\delta_{r}(\cC)$ if and only if
$\cC$ reaches the corresponding bounds for $d_H$ and $r$ is sufficiently closed to 1. This means that all
optimal CN codes come from the usual optimal codes. 
Secondly we present several constructions of CN codes with nice and flexible parameters $(n,K, \delta_r(\cC))$ by using bent functions.
\end{abstract}

\maketitle



\section{Introduction}\label{sec-one}
Shannon's work in 1948 pioneered two areas of research : information theory and mathematical coding theory.
Information theory has had a strong influence on the theoretical neuroscience (\cite{Attick},\cite{BT} et al.),
 the ideals in mathematical coding theory have received lots of attention more than decade ago and motivated by
the study of neurons called place cells. The discovery of the place cells by O$^\prime$keefe and Dostrovsky
\cite{OD} was a major breakthrough that led to a  shared 2014 Nobel Prize in Medicine or Physiology for O$^\prime$keefe.
A place cell encodes spatial information about an organism's surrounding by firing precisely when the organism is in the
corresponding place field. In this context,  a codeword represents the neural firing patten that occurs when the organism is
in the corresponding region of its environment. The set of codewords is called a neural code.

Place fields are modeled by several convex open sets which may be overlapped to each other.
The neural codes have been researched by the means of algebraic and combinatorial methods
(\cite{CFS},\cite{CIVY},\cite{GOY}).
Later, C. Curto et al. \cite{CIMRW} initiated the study of neural codes in more closer connection
with the classical error-correcting codes by discreting the convex place fields
(or called receptive fields in  \cite{CIMRW}).

Consider $n$ neurons $\{ 1,2,\cdots,n\}$ in brain. Each neuron $i$ has its
receptive fields $ \gS_i, X=\bigcup^n\limits_{i=1} \gS_i.$ For a stimulus $x \in X,$
we get a codeword $c(x)=(x_1, x_2,\cdots, x_n) \in \F^n_2, $ where for $1 \leq i \leq n,$
$$
x_i=  \left \{
\begin{array}{ll}
1, & \mbox{if} \ x_i \in  \gS_i , \\
0, & \mbox{if} \ x_i \not\in  \gS_i,
\end{array}
\right.
$$
$x_i \in  \gS_i$ means the neuron  $i$ fired by the stimulus  $x$,
$\ x_i \not\in  \gS_i$ means the neuron $i$  does not fire by the  stimulus $x$.
In Figure 1, $x$ is a stimulus, the codeword of $x$ is $c(x)=(00110)$.

Zero codeword $c=(0,0,\cdots, 0)$ means that no stimulus acts on $X$ at all.
For a set $\Sigma$ of $K$ stimuli ($K=|\Sigma|),$ we get a subset
$\cC=\{c(x):x \in \Sigma\}$ of $\F^n_2$ with $K$ codewords, called a combinatorial neural code (CN code for short).

\begin{figure}[h]
	\centering
	\includegraphics[scale=0.5]{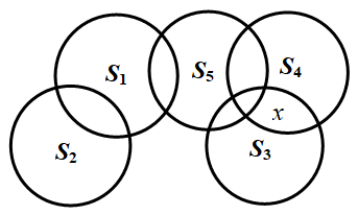}
 \centering
 \caption{ The codeword of $x$ is $c(x)=(00110)$}
   \end{figure}

In the real case, the ``receptive" vector
$c(x)^{\prime}=(x^{\prime}_1, x^{\prime}_2,\cdots,x^{\prime}_n) \in \F^n_2$ of the stimulus $x$
may have ``discrepancy" with $c(x)=(x_1, x_2,\cdots,x_n).$ The ideal model in \cite{CIMRW} is that

(I) $x_i =1$ and $x^{\prime}_i =0$ with probability $q < \frac{1}{2}$ which means that $x \in \gS_i$
but neuron $i$ does not fire with probability $q.$

(II) $x_i =0$ and $x^{\prime}_i =1$ with probability $p$ and $0 \leq p \leq q$ which means that $x \not\in \gS_i$
but neuron $i$ is fired with probability $p.$

From $c(x)$ to $c(x)^{\prime}$, we may consider the transmission in the following binary
asymmetric memoryless channel.

\begin{figure}[h]
	\centering
	\includegraphics[scale=0.5]{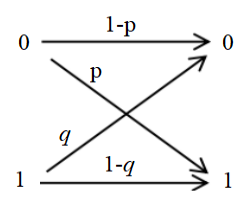}
\begin{center}\caption{Binary asymmetric memoryless channel}\end{center}
\end{figure}
For $p=q,$ the channel is symmetric. But usually $p$ is smaller than $q$ in the neuro science. It is hope to recover
the codeword $c(x)$ (and the acting stimulus $x$ ) from the received $c(x)^{\prime}.$
This leads to consider a suitable metric in the space $\F^n_2$ with respect to the above asymmetric channel.
Such problem has been researched as early as in 1970-1980's (\cite{CR},\cite{FLX}, and \cite{Kl}).
Further researches on ``Matched " metric (\cite{FW},\cite{OD},\cite{Qu}) to focus the ability on the usual decoding algorithms
(the maximum likehood decoding and nearest neighbour decoding ) for CN codes. In \cite{CR},
the authors suggested and discussed two kinds of metrics for CN codes. In this paper, we consider one of them.

Let $0 \leq p \leq q < \frac{1}{2}$ and $r=\log_{\frac{q}{1-p}}(\frac{p}{1-q}).$ It is easy to see that $r \geq 1$
and $r=1$ if and only if $p=q$ (the symmetric channel).
\begin{definition}(\cite{CR})
For $y=(y_1, y_2, \cdots, y_n)$ and  $x=(x_1, x_2, \cdots, x_n) \in \F^n_2,$ let
\begin{eqnarray*}
d_{10}(y,x)&=& \sharp \{ i  \mid 1 \leq i \leq n, (y_i, x_i)=(1,0)\} \\
d_{01}(y,x)&=& \sharp \{ i  \mid 1 \leq i \leq n, (y_i, x_i)=(0,1)\}.
\end{eqnarray*}
The (asymmetric) discrepancy between $y$ and $x$ is defined by
$$\delta(y,x)=\delta_{r}(y,x)=r d_{10}(y,x)+d_{01}(y,x).$$
\end{definition}
For a combinatorial neural code $\cC \in \F^n_2,$ the fundamental parameters of $\cC$ is $(n,K, \delta(\cC)),$
where $K=|\cC|\geq 2$ is the size of $\cC$ and
$$\delta(\cC))=\delta_{r}(\cC)=\min \{ \delta_{r}(c,c^{\prime}):c, c^{\prime} \in \cC, c \neq c^{\prime} \},$$
called the minimum discrepancy of $\cC.$

For $p=q,$ we get $r=1$ and
$\delta_{1}(y,x)=d_{10}(y,x)+d_{01}(y,x)=d_H(y,x)$ is the Hamming distance.
Therefore $\delta_{1}(\cC)=d_H(\cC),$ the usual minimum Hamming distance of $\cC.$
It is easy to see that $\delta(y,x) \geq 0$ and $\delta(y,x) = 0$ if and only if $y=x.$
If $p < q \ (r >1), \delta_{r}(x,y)$ and $\delta_{r}(y,x)$ may be different in general, but $\delta_{r}$
satisfies the triangular inequality (see \cite{CR}, Lemma 2.14).

In the last section of \cite{CR}, the authors raise an open problem: how to construct families of CN codes
with large size $K=|\cC|$ and minimum discrepancy $\delta_{r}(\cC)$ simultaneously.
In order to judge the goodness of a CN code, several classical bounds (Hamming, Singleton and Plotkin bounds)
of binary codes for $d_H$ are generalized to the ones for $\delta_{r}$ (see \cite{CR}, Lemma 4.7).

The aim of this paper is twofold. Firstly, in Section \ref{sec-two} we show that a CN code $\cC$
reaches the Hamming , Singleton or Plotkin bound for $\delta_{r}$ if and only if $\cC$ reaches the
corresponding bounds for $d_H$ and $r$ is sufficiently closed to 1 (so that $\delta_{r}$ is closed to $d_H=\delta_{1}$).
Therefore the optimal CN codes should be optimal as usual error-correcting code for $d_H.$
The optimal binary codes reaching
one of the above three bounds for usual Hamming distance $d_H$ have been found and their parameters $(n,K,d_H)$ are very limited.
Next, in Section \ref{sec-three},
we present a general construction of binary codes $\cC$
by the valuations of boolean functions with $m$ variables
at a subset $S$ of $\F^m_2,$ and show a formula on $\delta_{r}(\cC)$ in terms of such boolean functions.
This construction is given in Ding's paper \cite{Ding2} and since then many nice binary codes for $d_H$
have been found (see the survey paper \cite{LM}).
In this paper, we show three constructions for bent function
and determine the value of $\delta_{r}(\cC)$ for such code $\cC.$
 Section \ref{sec-four} is conclusion.

\section{A Characterization of optimal Combinatorial Neural Codes}\label{sec-two}

Firstly we introduce some bounds of the CN  codes
given in \cite{CR}.

\begin{lemma}(\cite{CR}, Lemma 4.7)\label{thm-bounds}
Let $ 0\leq p \leq q< \frac{1}{2}, r=\log_{\frac{q}{1-p}} (\frac{p}{1-q}) \geq 1$ and
$\cC$ be a CN code with parameters $(n,K,\delta(\cC)), k=\log_{2}K$ and $\delta(\cC)=\delta_{r}(\cC).$
We have
\begin{itemize}
\item[(1)]\ (Singleton bound) $k \leq n-\lceil \frac{2\delta(\cC)}{r+1}\rceil +1,$
where $\lceil A \rceil$ is the minimal integer $a$ such that $a \geq A;$

\item[(2)]\ (Hamming bound) $K \leq 2^{n}/ \sum^T\limits_{i=0}\binom{n}{i}, T=\lfloor \frac{\delta(\cC)}{r+1} \rfloor, $
where $\lfloor A \rfloor$ is the maximal integer $a$ such that $a < A;$

\item[(3)]\ (Plotkin bound) if $d=\lceil \frac{2\delta(\cC)}{r+1} \rceil$ satisfies $2d > n,$
then $K \leq [\frac{2d}{2d-n}],$ where $[A]$ is the maximal integer $a$ such that $a \leq A.$
\end{itemize}
\end{lemma}

If $p=q,$ then $r=1$ and $\delta(\cC)=\delta_{1}(\cC)=d_{H}(\cC)$ is the minimal Hamming
distance of $\cC.$ Lemma \ref{thm-bounds} presents the bounds of $\cC$ with respect to $d_{H}.$

\begin{corollary}\label{thm-main}
Let $\cC$ be a binary code with parameters $(n,K,d_{H}(\cC)),k=\log_{2}K.$ We have
\begin{itemize}
\item[(1)]\  (Singleton bound) $k \leq n-d_{H}(\cC)+1;$

\item[(2)]\ (Hamming bound) $K \leq 2^n / \sum^t\limits_{i=0}\binom{n}{i},$ where
$t=\lfloor \frac{d_{H}(\cC)}{2}\rfloor=[\frac{d_{H}(\cC)-1}{2}];$

\item[(3)]\ (Plotkin bound) if $2d_{H}(\cC)>n,$ then $K \leq [\frac{2d_{H}(\cC)}{d_{H}(\cC)-n}].$
\end{itemize}
\end{corollary}

 In this section we show that for a binary code $\cC,$ the parameters of $\cC$ reaches the Singleton
 bound, Hamming bound or Plotkin bound for $\delta_{r}(\cC)$ if and only if $\cC$ reaches
 the corresponding bound for $d_{H}(\cC)$ and $r$ is sufficiently close to 1. Firstly we need some results in
 \cite{CR} which are used in this paper.

\begin{lemma}\label{thm-cr}
Let $\cC$ be a CN code with parameters $(n,K,\delta_{r}(\cC)).$ We have
\begin{itemize}
\item[(1)]\ (\cite{CR}, \ Proposition 4.3) $d_{H}(\cC) \leq \delta_{r}(\cC) \leq \frac{r+1}{2}d_{H}(\cC);$

\item[(2)]\ (\cite{CR}, \ Proposition 4.8) if $\cC$ is a linear code, then $\delta_{r}(\cC)=d_{H}(\cC).$
\end{itemize}
\end{lemma}

Now we state the main results in this section.
\begin{theorem}\label{thm-six}
Let $\cC$ be a CN code with parameters $(n,K,\delta_{r}(\cC)).$ Then
\begin{itemize}
\item[(1)] \ $\cC$ reaches the Singleton or Plotkin bound for $\delta_{r}$ if and only if
$\cC$ reaches the corresponding bound for $d_{H}$ and
$\delta_{r}(\cC) > \frac{r+1}{2}(d_{H}(\cC)-1).$
Particularly, if $r < \frac{d_{H}(\cC)+1}{d_{H}(\cC)-1},$ then $\cC$ reaches the
Singleton or Plotkin bound for $\delta_{r}$ if and only if $\cC$ reaches the corresponding bound
for $d_{H}.$

\item[(2)] \ $\cC$ reaches the Hamming bound for $\delta_{r}$ if and only if
$\cC$ reaches the Hamming bound for $d_{H}$ (perfect code), $d_{H}=2t+1$ is odd,
and $\delta_{r}(\cC) > t(r+1).$ Particulary, if $r < \frac{t+1}{t}$ and
$d_{H}(\cC)=2t+1,$ then $\cC$ reaches the Hamming bound for $\delta_r$ if and only if $\cC$ is a perfect code.
\end{itemize}
\end{theorem}

\begin{proof}
(1) Let $k=\log_{2}K.$ If $\cC$ reaches the Singleton bound for $\delta_{r},$ then
$k=n-\lceil \frac{2 \delta_{r}(\cC)}{r+1}\rceil +1.$
Since $k\leq n-d_{H}(\cC)+1,$
we get $\lceil \frac{2 \delta_{r}(\cC)}{r+1} \rceil \geq d_{H}(\cC).$ But
$d_{H}(\cC) \geq  \lceil \frac{2\delta_{r}(\cC)}{r+1}\rceil$ by Lemma \ref{thm-cr} (1),
we get $d_{H}(\cC) = \lceil \frac{2\delta_{r}(\cC)}{r+1}\rceil$ and $k=n-d_{H}(\cC)+1$
which means that $\cC$ reaches the Singleton bound for $d_{H}.$ Moreover,
\begin{eqnarray*}
\lceil \frac{2 \delta_{r}(\cC)}{r+1} \rceil =d_{H}(\cC)
& \Leftrightarrow  & d_{H}(\cC)-1 <\frac{2 \delta_{r}(\cC)}{r+1} \leq d_{H}(\cC) \\
& \Leftrightarrow & d_{H}(\cC)-1 < \frac{2 \delta_{r}(\cC)}{r+1} \\
 (\mbox{since} \ \
 \frac{ 2 \delta_{r}(\cC)}{r+1} & \leq & d_{H}(\cC) \ \mbox{by Lemma } \ref{thm-cr} \ (1))\\
& \Leftrightarrow & \delta_{r}(\cC) > \frac{r+1}{2}(d_{H}(\cC)-1).
\end{eqnarray*}

Therefore $\cC$ reaches the Singleton bound for $\delta_{r}$ if and only if $\cC$ reaches the Singleton bound for $d_{H}$
and $$\delta_{r}(\cC) > \frac{r+1}{2}(d_{H}(\cC)-1).$$

If $\cC$ reaches the Plotkin bound for $\delta_{r}$, then $2d=2\lceil \frac{2 \delta_{r}(\cC)}{r+1} \rceil > n$
and $K=\left[\frac{2d}{2d-n}\right].$ From $d_{H}(\cC) \geq \frac{2 \delta_{r}(\cC)}{r+1},$ we know that
$d_{H}(\cC) \geq \lceil \frac{2 \delta_{r}(\cC)}{r+1} \rceil$ and
then $$2d_{H}(\cC) \geq 2 \lceil \frac{2 \delta_{r}(\cC)}{r+1} \rceil=2d>n.$$
Furthermore, from $d_{H}(\cC) \geq d >\frac{n}{2},$ we get $\frac{2d_{H}(\cC)}{2d_{H}(\cC)-n} \leq \frac{2d}{2d-n}.$
Therefore
$$\left[\frac{2d_{H}(\cC)}{2d_{H}(\cC)-n} \right]\leq \left[\frac{2d}{2d-n} \right]=K.$$
 By the Plotkin bound for $d_{H},$ we get
$\left[\frac{2d_{H}(\cC)}{2d_{H}(\cC)-n} \right] =K$ and $d_{H}(\cC)=d.$ Thus $\cC$ reaches the Plotkin bound for $\delta_{r}$
if and only if $\cC$ reaches the Plotkin bound for $d_{H}$ and $d=d_{H}(\cC).$ Moreover,
\begin{eqnarray*}
d=d_{H}(\cC) & \Leftrightarrow  & \lceil \frac{2\delta_{r}(\cC)}{r+1}\rceil= d_{H}(\cC) \\
& \Leftrightarrow & \delta_{r}(\cC)> \frac{r+1}{2}(d_{H}(\cC)-1).
\end{eqnarray*}
This gives the first statement of (1). If $r < \frac{d_{H}(\cC)+1}{d_{H}(\cC)-1},$ then
$\frac{r+1}{2}(d_{H}(\cC)-1) < d_{H}(\cC) \leq \delta_{r}(\cC),$ and the last statement of (1) is true.

(2) If $\cC$ reaches the Hamming bound for $\delta_{r}$, then $K=2^n / \sum^{T}\limits_{i=0}\binom{n}{i},$
where $T=\lfloor \frac{\delta_{r}(\cC)}{r+1}\rfloor.$ Namely, $T < \frac{\delta_{r}(\cC)}{r+1} \leq T+1.$
By Lemma \ref{thm-cr} (1), we get $\frac{\delta_{r}(\cC)}{r+1} \leq \frac{d_{H}(\cC)}{2}$
and $T \leq \lfloor \frac{d_{H}(\cC)}{2}\rfloor=[\frac{d_{H}(\cC)-1}{2}].$ From the Hamming bound for $d_{H}$ we know that
$$2^n / \sum^{t}\limits_{i=0}\binom{n}{i} \geq K=2^n / \sum^{T}\limits_{i=0}\binom{n}{i}$$
where $t=[\frac{d_{H}(\cC)-1}{2}] \geq T.$ Therefore $t=T$ and $\cC$ reaches the Hamming bound for $d_{H}.$
It is well-known that for such code, $d_{H}(\cC)=2t+1$ should be odd.
Moreover,
\begin{eqnarray*}
T=t & \Leftrightarrow  & t < \frac{\delta_{r}(\cC)}{r+1} \leq t+1\\
& \Leftrightarrow & t < \frac{\delta_{r}(\cC)}{r+1} \\
 & \Leftrightarrow & \delta_{r}(\cC) >t(r+1).
\end{eqnarray*}
The third statement holds since by Lemma  \ref{thm-cr}  (1),
$\frac{\delta_{r}(\cC)}{r+1} \leq \frac{d_{H}(\cC)}{2}=\frac{2t+1}{2}< t+1.$
This gives the first statement of (2). Moreover, if $r <\frac{t+1}{t},$ then
$\delta_{r}(\cC) \geq d_{H}(\cC)=2t+1 >t(r+1),$ and the last statement of (2) is true.
\end{proof}

Theorem \ref{thm-six} shows that all binary codes reaching the Singleton,
Hamming or Plotkin bound for $\delta_{r}(\cC)$ should be the ones reaching the corresponding bound for $d_{H}.$
The parameters of such binary codes are very limited (see \cite{MS}). The parameters $[n,k,d_{H}]$ of binary codes $\cC$ are

$\ast$ two trivial cases $[n,n,1]$ and $[n,n-1,2] \ (n \geq 2)$ for $\cC$ reaching the Singleton bound;

$\ast$ $[2^m -1, 2^m -m-1,3] \ (m \geq 2, \ \mbox{Hamming codes})$ and $[23,12,7]$ (Golay code) for $\cC$ reaching the Hamming bound.

Any binary code $\cC$ reaching the Plotkin bound should satisfy the strong condition $2 d_{H}(\cC)> n$ and any different codewords
of $\cC$ have the same Hamming distance. The typical examples of such codes are made by using the Hadamard matrices.

\begin{definition}
For $n\geq 2,$ an $n \times n$ matrix $ H=(h_{ij})_{0 \leq i,j \leq n-1}$
is called a Hadamard matrix if $h_{ij} \in \{\pm 1\} \  (0 \leq i,j \leq n-1)$ and
$$HH^T=n I_n \  \ \ (I_n \ \mbox{is the identity matrix of order } n ).$$
\end{definition}
It is well-known that for any Hadamard matrix of order $n \geq 3, n$ should be divided by $4.$
One of the famous conjecture in combinatorial theory is that for all $n \equiv 0 \  (\bmod ~4)$ and $n \geq 4,$
there exists a Hadamard matrix of order $n.$ This conjecture is verified for all $n=4m$ up to a very large number.

Let $ H=(h_{ij})_{0 \leq i,j \leq n-1}$ be a Hadamard matrix of order $n=4m.$ We can assume that $h_{i0}=1 \ (0 \leq i \leq n-1)$
without losing of generality. Namely,

$$H=\left[\begin{array}{cc}
	1 & \\
	\vdots  & H^{\prime}  \\
	1 &
\end{array}
\right],~H^{\prime}=(h_{ij}) \ (0 \leq i \leq n-1, 1 \leq j \leq n-1).$$

Let $h_{ij}=(-1)^{c_{ij}}, c_{ij} \in \F_2=\{0,1\}$ and
$$c_i=(c_{i1},c_{i2},\cdots, c_{i,n-1}) \in \F^{n-1}_2 \ \ (0 \leq i \leq n-1).$$
Consider the binary code $\cC=\{c_0,c_1,\cdots,c_{n-1}\}, K=|\cC|=n.$ From
$nI_n=HH^T,$ we can see that $d_{H}(c_i, c_j)=\frac{n}{2}$ for all $0 \leq i \neq j \leq n-1.$
Therefore the parameters of $\cC$ is $(n-1, K=n, d_{H}(\cC)=\frac{n}{2}).$
From $\frac{2 d_{H}(\cC)}{2d_{H}(\cC)-(n-1)}=n=K,$ we know that $\cC$ reaches the Plotkin bound for $d_{H}.$

We will show other examples in  next  section (see remark of Theorem \ref{thm-caseA}).

Beside the optimal CN codes which have  limited parameters, it is nature to ask
how to construct the CN codes $\cC$ with more flexible parameters $(n,K,\delta_{r}(\cC)),$
and larger size  $K$ and $\delta_{r}(\cC)$. We will consider this problem in next  section.

\section{Construction of CN codes by boolean functions}\label{sec-three}

Inspired by the constructions of error-correcting codes suggested by Ding \cite{Ding1},\cite{Ding2}.
We consider the following construction of CN codes.

Let $m \geq 2$ and $\B_m$ be the set of boolean functions with $m$ variables
$$f(x)=f(x_1, x_2,\cdots, x_m): \F^m_2 \longrightarrow \F_2.$$
$\B_m$ is a commutative ring and
$$\B_m=\F_2[x_1, x_2,\cdots, x_m]/ (x^2_1-x_1, \cdots, x^2_m-x_m),$$
then $|\B_m|=2^{2^m}.$

For each subset $V \subseteq \F^m_2,$ we have an ideal of the ring $\B_m:$
$$I(V)=\{f(x) \in \B_m:f(x)=0 \ \mbox{for all } x \in V \}.$$
Let $|V|=n$ and $V=\{x_1, x_2,\cdots, x_n\} \ (2 \leq n \leq 2^m).$ We have a mapping
$$\varphi_V :\B_m \longrightarrow \F^n_2,$$
$$ \varphi_V(f)=c_f=(f(x_1), f(x_2),\cdots, f(x_n)).$$

Then for each subset $S \subset \B_m,$ we get a binary code in $\F^n_2,$
$$\cC=\cC(V,S)=\varphi_V(S)=\{\varphi_V(f) \in \F^n_2: f \in S\}.$$
We hope that such binary code $\cC=\cC(V,S)$ have nice and flexible parameters
$(n,K, \delta_{r}(\cC))$ by choosing suitable subset $V \subseteq \F^m_2$ and $S \subset \B_m.$
~\\
\begin{theorem}\label{thm-boolean}
For the binary code $\cC=\cC(V,S) \subseteq \F^n_2$ defined above, the length of codewords is $n=|V|$ and
\begin{itemize}
\item[(1)] the size $K=|\cC|$ is $|S|$ if and only if for each pair $(f,g), f,g \in S, f \neq g,$
 we have $ f \not\equiv g \  (\bmod ~I(V))$ (which means that $f-g \not\in I(V)$).

\item[(2)] for $f,g \in S, f \neq g,$
\begin{eqnarray*}\delta_{r}(c_f, c_g)&=&
 \frac{r+1}{4}\left(n-\sum\limits_{x \in V}(-1)^{f(x)+g(x)}\right) \\
 &+& \frac{r-1}{4}\left(\sum\limits_{x \in V}(-1)^{g(x)}-\sum\limits_{x \in V}(-1)^{f(x)}\right).\end{eqnarray*}

\item[(3)] $\delta_{r}(\cC)
=\mbox{min}\{\delta_{r}(c_f, c_g) \mid f,g \in S, f \neq g,
\sum\limits_{x \in V}(-1)^{g(x)} \leq \sum\limits_{x \in V}(-1)^{f(x)}\}.$
\end{itemize}
\end{theorem}
~\\
\begin{proof}
(1) Consider the mapping
$$\varphi_V:\B_m \rightarrow \F^n_2, \varphi_V(f)=c_f=(f(x))_{x \in V}. $$
This is an $\F_2$-linear mapping and $\mbox{ker}(\varphi_V)=I(V).$ Therefore
\begin{eqnarray*}|\cC|=|S| &\Leftrightarrow & \varphi_V :S \rightarrow \F^n_2 \
\mbox{is injective} \\ &\Leftrightarrow &  \mbox{for} \ f,g \in S, \ f \neq g,\  \mbox{we have} \  c_f \neq c_g.
\end{eqnarray*}
But
$$c_f \neq c_g \Leftrightarrow c_{f-g}\neq 0 \Leftrightarrow f-g \not\in I(V).$$
This completes the proof of (1).

(2) For $f,g \in \B_m,$
\begin{eqnarray*}
&&d_{10}(c_f, c_g) \\
& = & \sharp \{x\in V:f(x)=1, g(x)=0\} \\
&=&\frac{1}{4}\sum\limits_{x \in V}\left(1-(-1)^{f(x)})(1+(-1)^{g(x)}\right)\\
& =&\frac{1}{4}\left(|V|- \sum\limits_{x \in V}(-1)^{f(x)+g(x)}\right) \\
&&+\frac{1}{4}\left(\sum\limits_{x \in V}(-1)^{g(x)}-\sum\limits_{x \in V}(-1)^{f(x)}\right).\end{eqnarray*}
\begin{eqnarray*}
d_{01}(c_f, c_g) &=&  d_{10}(c_g, c_f)\\
&=&\frac{1}{4}\left(|V|- \sum\limits_{x \in V}(-1)^{f(x)+g(x)}\right)\\
&&+\frac{1}{4}\left(\sum\limits_{x \in V}(-1)^{f(x)}-\sum\limits_{x \in V}(-1)^{g(x)}\right).
\end{eqnarray*}
Therefore
\begin{eqnarray*}
\delta_{r}(c_f, c_g)& = & r d_{10}(c_f, c_g)+d_{01}(c_f, c_g) \\
& =&\frac{r+1}{4}\left(n- \sum\limits_{x \in V}(-1)^{f(x)+g(x)}\right)\\
&&+ \frac{r-1}{4}\left(\sum\limits_{x \in V}(-1)^{g(x)}-\sum\limits_{x \in V}(-1)^{f(x)}\right).\\
\end{eqnarray*}

(3) Let $\delta_{r}(c_f, c_g)=r d_{10}(c_f, c_g)+d_{01}(c_f, c_g)=\delta_{r}(\cC),$
then
$$\delta_{r}(\cC) \leq \delta_{r}(c_g, c_f)=r d_{01}(c_f, c_g)+d_{10}(c_f, c_g).$$
Therefore $(r-1)(d_{10}(c_f, c_g)-d_{01}(c_f, c_g)) \leq 0.$ If $r > 1,$ then
$d_{10}(c_f, c_g) \leq d_{01}(c_f, c_g),$ which means that
$\sum\limits_{x \in V}(-1)^{g(x)} \leq \sum\limits_{x \in V}(-1)^{f(x)}$.
If $r=1,$ then $\delta_{r}(c_f, c_g)=\delta_{r}(c_g, c_f)=d_{H}(c_g, c_f)$
and we can choose $f$ and $g$ such that $\sum\limits_{x \in V}(-1)^{g(x)} \leq \sum\limits_{x \in V}(-1)^{f(x)}.$
This completes the proof of (3).
\end{proof}

\vspace{0.3cm}
After the Ding's works (\cite{Ding1},\cite{Ding2}), a huge papers appear to determine the parameters
$(n, K, d_{H}(\cC))$ for binary codes constructed by many types of boolean functions
(see the survey paper \cite{LM}).
In this paper we just consider a simple case by taking bent functions and determine the minimum discrepany $\delta_{\gamma}(\cC)$
for such binary codes $\cC.$ For bent functions, the interested readers may refer to the book \cite{Sen} for more details.

\begin{definition}
For a boolean function $f(x) \in \B_{m} \ (m \geq 2),$ the Walsh transformation of $f(x)$
is the function $$W_{f}(y)=\sum\limits_{x \in \F^m_2} (-1)^{f(x)+x\cdot y},$$
where for $x=(x_1,x_2,\cdots, x_m),y=(y_1,y_2,\cdots, y_m) \in \F^m_2,$
$x\cdot y = \sum^m\limits_{i=1} x_i y_i \in \F_2.$
\end{definition}

$f(x) \in \B_{m}$ is called a bent function if for all
$y \in \F^m_2, W_{f}(y)=2^{\frac{m}{2}}$ or $-2^{\frac{m}{2}}.$
It is well known that for all even $m, m \geq 2,$ there exists bent function $f(x)$ in $\B_{m}$.
We use the following properties on the bent function.

\begin{lemma}\label{thm-bent}
Let $f(x) \in \B_{m}$ be a bent function, $m=2k. $

(1) For each $0 \neq a \in \F^m_2, \sum\limits_{x \in \F^m_2} (-1)^{f(x+a)+f(x)}=0.$

(2) Let $D_f=\{ x \in \F^m_2: f(x)=1\}$ be the support of $f,$ then
$|D_f|=2^{m-1}+\varepsilon 2^{k-1}$ and $W_{f}(0)=\sum\limits_{x \in \F^m_2} (-1)^{f(x)}=- \varepsilon 2^{k},$
where $\varepsilon = \pm 1.$ If $|D_f|=2^{m-1}+ 2^{k-1},$ then $f(x)+1$ is also a bent function and
$|D_{f+1}|=2^{m-1}- 2^{k-1}.$
\end{lemma}

Now we present three constructions of CN codes by boolean functions.
\vspace{0.15cm}
\textbf{Construction A}

Let $f(x) \in \B_{m}$ be a bent function, $m=2k \neq 4.$ For $a \in \F^m_2, b\in \F_2,$
let $f_{a,b}(x)=f(x+a)+b \in \B_{m}$ (which is also a bent function). Consider the following binary code
$$\cC=\cC_f=\{c(a,b)=(f_{a,b}(x))_{x  \in \F^m_2} \mid  a \in \F^m_2, b\in \F_2\} \subseteq \F^n_2, $$
where $ n=2^m.$
~\\
\begin{theorem}\label{thm-caseA}
The code $\cC=\cC_f$ defined above has parameters
$(n=2^m, K=2^{m+1},\delta_{r}(\cC)=\frac{r+1}{4}2^m - \frac{r-1}{4}2^{k+1}).$
\end{theorem}

\begin{proof}
It is obvious that $n=2^m$ and $K \leq \sharp \{(a,b): a \in \F^m_2, b\in \F_2\}=2^{m+1}.$
Now we compute $\delta_{r}(\cC)$ by Theorem \ref{thm-boolean} and Lemma \ref{thm-bent}.
For codewords $c(a,b)$ and $c(a^{\prime},b^{\prime}), \ c(a,b) \neq c(a^{\prime},b^{\prime}),$
\begin{eqnarray*}
&& \delta_{r}(c(a,b), c(a^{\prime},b^{\prime}))\\
&=&\frac{r+1}{4}\left(n-\sum\limits_{x \in \F^m_2} (-1)^{f_{a,b}(x)+f_{a^{\prime},b^{\prime}}(x)}\right)\\
&&+\frac{r-1}{4}\left(\sum\limits_{x \in \F^m_2}(-1)^{f_{a^{\prime},b^{\prime}}(x)}-\sum\limits_{x \in \F^m_2}(-1)^{f_{a,b}(x)}\right)\\
&=&\frac{r+1}{4}\left(2^m-\sum\limits_{x \in \F^m_2} (-1)^{f(x+a)+f(x+a^{\prime})+b+b^{\prime}}\right)\\
&&+\frac{r-1}{4}\left(\sum\limits_{x \in \F^m_2}(-1)^{f(x+a^{\prime})+b^{\prime}}-\sum\limits_{x \in \F^m_2}(-1)^{f(x+a)+b}\right)\\
&=& \frac{r+1}{4}\left(2^m-(-1)^{b+b^{\prime}}\sum\limits_{x \in \F^m_2}(-1)^{f(x+a)+f(x+a^{\prime})}\right)\\
&&+\frac{r-1}{4}\left( ((-1)^{b^{\prime}}-(-1)^{b})\sum\limits_{x \in \F^m_2}(-1)^{f(x)}  \right)\\
&=&\frac{r+1}{4}\left(2^m-(-1)^{b+b^{\prime}} 2^m \delta_{a,a^{\prime}}\right)\\
&&+\frac{r-1}{4}\left(((-1)^{b^{\prime}}-(-1)^b)(-\varepsilon 2^k) \right)
\end{eqnarray*}
where $\delta_{a,a^{\prime}}=1$ if $a=a^{\prime}$ and $\delta_{a,a^{\prime}}=0$ otherwise. Therefore for
$(a,b) \neq (a^{\prime},b^{\prime}),$
\begin{eqnarray*}
&& \delta_{r}(c(a,b), c(a^{\prime},b^{\prime}))\\
&= & \left \{
\begin{array}{ll}
\frac{r+1}{4}2^m+\frac{r-1}{4} ((-1)^b -(-1)^{b^{\prime}})\varepsilon 2^k, & \mbox{if} \ a \neq a^{\prime}, \\
\frac{r+1}{4}2^{m+1}+\frac{r-1}{4}(-1)^b \varepsilon 2^{k+1}, & \mbox{if} \scriptstyle  a = a^{\prime} \mbox{and}
\atop \scriptstyle  b + b^{\prime}=1.
\end{array}
\right.
\end{eqnarray*}
By Theorem \ref{thm-boolean} (3), the coefficient of $\frac{r-1}{4}$ in the right-hand side should not be positive.
Then we can  see that the minimal value of the right-hand side is
$$\delta_{r}(\cC)=\frac{r+1}{4}2^m-\frac{r-1}{4}2^{k+1}, $$
and $d_{H}(\cC)=\delta_{1}(\cC)=2^{m-1},\delta_{r}(\cC)-d_{H}(\cC)=\frac{r-1}{4}(2^m -2^{k+1}) > 0$ for $r >1.$
Moreover, for $(a,b) \neq (a^{\prime},b^{\prime}), d_{H}(c(a,b),c(a^{\prime},b^{\prime})) \geq 2^{m-1} >0.$ We get $K=|\cC|=2^{m+1}.$
\end{proof}

\begin{remark}
The parameters of $\cC$ is $(n=2^m, K=2^{m+1}, d_H(\cC)=2^{m-1}).$ For $\alpha \in \F_2,$ let
$$\cC^{(\alpha)}=\{(c_1,c_2,\cdots,c_{n-1}) | (\alpha,c_1,c_2,\cdots,c_{n-1}) \in \cC \}.$$
If $c(a,b) =(c_0,c_1, \cdots, c_{n-1}),$ then $c(a,b+1) =(c_0 +1,c_1 +1, \cdots, c_{n-1} +1).$ Therefore $|\cC^{(0)}|=|\cC^{(1)}|=2^m.$
Namely, the parameters of $\cC^{(\alpha)}$ is $(n=2^m -1, K=2^{m}, d_H=d_H(\cC)=2^{m-1})$ and
$\delta_{r}(\cC^{(\alpha)})=\delta_{r}(\cC)$ for $\alpha=0$ and $1.$ By $\frac{2 d_H}{d_H-n}=2^m=K,$
we know that both of $\cC^{(0)}$ and $\cC^{(1)}$ reaches the Plotkin bound for $d_H.$ By Theorem \ref{thm-six} (1),
$\cC^{(0)}$ and $\cC^{(1)}$ reaches the Plotkin bound for $\delta_{r}$ if and only if

\begin{eqnarray*}
\frac{r +1}{4} 2^m -\frac{r -1}{4} 2^{k+1}
&=& \delta_{r}(\cC^{(\alpha)}) > \frac{r +1}{2}(d_H(\cC^{(\alpha)})-1)\\
&=&\frac{r +1}{2}(2^{m-1}-1),\end{eqnarray*}
which means that $r < \frac{2^k +1}{2^k -1}.$
\end{remark}
\vspace{0.15cm}
\textbf{Construction B}

Let $f(x) \in \B_{m}$ be a bent function, $m=2k \neq 4, D_f=\{x \in \F^m_2 | f(x)=1\}$
be the support of $f, n=|D_f|=2^{m-1}+\varepsilon 2^{k-1},  \varepsilon  \in \{ \pm 1\}.$
For $a \in \F^m_2, b\in \F_2,$ we get $c(a,b)=(ax+b)_{x \in D_f} \in \F^n_2.$ Then
$$\cC=\cC(f)=\{c(a,b) \mid a \in \F^m_2, b\in \F_2 \} \subseteq \F^n_2$$
is a linear code and $\delta_{r}(\cC)=d_H(\cC)$ for all $r \geq 1$ (Lemma \ref{thm-cr} (2)).
The following result shows that if we omit the zero codewords $c(0,0)=0$ from $\cC,$ then for $ \cC^{\prime}=\cC \backslash \{0\},$
the value of $\delta_{r}(\cC^{\prime})$ can be larger than $d_H(\cC^{\prime})$ when $r >1.$
~\\
\begin{theorem}\label{thm-caseB}
Let $f(x) \in \B_{m}$ be a bent function, $m=2k \geq 4, |D_f|=2^{m-1}+\varepsilon 2^{k-1}, \varepsilon \in \{\pm 1 \}$
and

$\cC^{\prime}=\cC^{\prime}(f)$
$=\{c(a,b)=(a \cdot x +b)_{x \in D_f} | a \in \F^m_2, b\in \F_2, (a,b) \neq (0,0)\}.$

Then the parameters of $\cC^{\prime}$ is
$(n=|D_f|, K=2^{m+1}-1, \delta_{r}(\cC^{\prime})=\frac{r+1}{4}(n-2^{k-1})-\frac{r -1}{4} (2^{m-1}-2^k)),$
$d_H(\cC^{\prime})=\frac{1}{2}(n-2^{k-1})$ and for $r=1, \delta_{r}(\cC^{\prime}) > d_H(\cC^{\prime}).$
\end{theorem}

\begin{proof}
By Theorem \ref{thm-boolean}, we have
\begin{eqnarray}\label{eqn-prime}
&& \delta_{r}(c(a,b), c(a^{\prime},b^{\prime}))\nonumber \\
&=&\frac{r+1}{4}\left(n-\sum\limits_{x \in D_f} (-1)^{a \cdot x +b +a^{\prime} \cdot x +b^{\prime}}\right)\nonumber \\
 &&+\frac{r-1}{4}\left(\sum\limits_{x \in D_f}(-1)^{a^{\prime} \cdot x +b^{\prime}}-\sum\limits_{x \in D_f}(-1)^{a \cdot x +b} \right)\nonumber \\
&=& \frac{r+1}{4} \left(n-(-1)^{b+b^{\prime}}\sum\limits_{x \in D_f}(-1)^{(a+a^{\prime})x} \right)\\
&+& \frac{r-1}{4} \left((-1)^{b^{\prime}}\sum\limits_{x \in D_f}(-1)^{a^{\prime}\cdot x}-(-1)^b \sum\limits_{x \in D_f}(-1)^{a \cdot x} \right)\nonumber.
\end{eqnarray}

For any $a \in \F^m_2, $
\begin{eqnarray*}&&\sum\limits_{x \in D_f}(-1)^{a \cdot x} \\&=&
\frac{1}{2}\sum\limits_{x \in \F^m_2}\left(1-(-1)^{f(x)}\right)(-1)^{a \cdot x}\\
&=&\frac{1}{2}\left(\sum\limits_{x \in \F^m_2}(-1)^{a \cdot x}-W_f(a)\right)\\
&=& \frac{1}{2}\left(2^m \delta_{a} -W_f(a)\right),
\end{eqnarray*}
where $\delta_{a}=1$ for $a=0$ and $\delta_{a}=0$ otherwise. Therefore
\begin{eqnarray*}
&&\delta_{r}(c(a,b), c(a^{\prime},b^{\prime}))\\
&=& \frac{r+1}{4} \left(n-(-1)^{b+b^{\prime}} \frac{1}{2} (2^m \delta_{a+a^{\prime}}-W_f(a+a^{\prime}))\right) \\
&+& \frac{r-1}{4}\left(\frac{1}{2} (-1)^{b^{\prime}}(2^m \delta_{a^{\prime}}-W_f(a^{\prime}))\right)\\
&+& \frac{r-1}{4}\left(-\frac{1}{2} (-1)^{b}(2^m \delta_{a}-W_f(a))\right).
\end{eqnarray*}
If $a \neq a^{\prime},$ then
\begin{eqnarray*}
&&\delta_{r}(c(a,b), c(a^{\prime},b^{\prime}))\\
&=& \frac{r+1}{4} \left(n+(-1)^{b+b^{\prime}} \frac{1}{2} W_f(a+a^{\prime})\right) \\
&& + \frac{r-1}{4}\left( \frac{1}{2} (-1)^{b^{\prime}}(2^m \delta_{a^{\prime}}-W_f(a^{\prime}))\right)\\
&&-\frac{r-1}{4}\left(\frac{1}{2} (-1)^{b}(2^m \delta_{a}-W_f(a))\right).
\end{eqnarray*}
The minimum value of the right-hand side is
$$A=\frac{r+1}{4} (2^{m-1}+\varepsilon 2^{k-1}-2^{k-1})+\frac{r-1}{4}(-2^{m-1}-2^k).$$
If $a = a^{\prime},$ then $b \neq b^{\prime},$ namely $(b, b^{\prime})=(0,1)$ or $(1,0).$
Then by assumption $(a,b), (a^{\prime},b^{\prime}) \neq (0,0),$ we have $a =a^{\prime} \neq 0. $ In this case, by (\ref{eqn-prime})
\begin{eqnarray*}
&&\delta_{r}(c(a,b), c(a^{\prime},b^{\prime})) \\&=& \frac{r+1}{4} (n+|D_f|)
+\frac{r-1}{4}(2 \cdot (-1)^{b^{\prime}} W_f(a)) \\
&=& \frac{r+1}{4} (2^m +\varepsilon 2^k)  +\frac{r-1}{4}(2 \cdot (-1)^{b^{\prime}} W_f(a)).
\end{eqnarray*}
The minimal value of the right-hand side is
$$B=\frac{r+1}{4}(2^m +\varepsilon 2^k)-\frac{r-1}{4} 2^{k+1}.$$ Since
$$B-A=\frac{r+1}{4}(2^{m-1}+(\varepsilon+1)2^{k-1})+\frac{r-1}{4}(2^{m-1}-2^k) >0.$$
We get $$\delta_{r}(\cC^{\prime})=A=\frac{r+1}{4}(2^{m-1}+(\varepsilon-1)2^{k-1})-\frac{r-1}{4}(2^{m-1}+2^k).$$
Moreover,$$d_H(\cC^{\prime})=\delta_{1}(\cC^{\prime})=\frac{1}{2}(2^{m-1}+(\varepsilon-1)2^{k-1})=\frac{1}{2}(n-2^{k-1})$$
and $\delta_{r}(\cC^{\prime})-d_{H}(\cC^{\prime})=\frac{r-1}{4}(\varepsilon+1)2^{k-1}.$
Namely, if $r>1$ and $\varepsilon=1,$ then $\delta_{r}(\cC^{\prime})> d_{H}(\cC^{\prime}).$
For $(a,b) \neq (a^{\prime},b^{\prime}),$
\begin{eqnarray*}
d_H(c(a,b),c(a^{\prime},b^{\prime})) & \geq & d_H(\cC^{\prime})
=\frac{1}{2}(2^{m-1}+(\varepsilon-1)2^{k-1}) \\
& \geq & \frac{1}{2}(2^{m-1}-2^{k}) >0
\end{eqnarray*}
which implies that $c(a,b) \neq c(a^{\prime},b^{\prime}).$ Therefore $K=|\cC^{\prime}|=2^{m+1}-1.$
This completes the proof of Theorem \ref{thm-caseB}.
\end{proof}
\vspace{0.15cm}
\textbf{Construction C}

Let $m=2k \geq 4.$ A subset $S$ of $\B_m$ is called bent set if for any distinct boolean functions
$f$ and $g$ in $S,$ $f+g$ is bent. It is proved in \cite{BM} that the maximal size $\mid S \mid$
of a bent set $S$ in $\B_m$ is $2^{m-1}.$ Such maximal bent set in $\B_m$ has been presented in \cite{BT}
for all even $m \geq 4$ as following.

We view $\F^{m-1}_2$ as the finite field $\F_{2^{m-1}},$ then $\F^{m}_2$ can be viewed as
$\F^{m}_2=\F_{2^{m-1}} \times \F_2$ with inner product
$$((x,x_m),(y,y_m))=\tr(xy)+x_m y_m \ \ x,y \in \F_{2^{m-1}}, x_m, y_m \in \F_2,$$
where $\tr:\F_{2^{m-1}} \longrightarrow \F_2$ is the trace mapping $\tr(x)=\sum^{m-2}_{\lambda=0}x^{2^{\lambda}}.$

Consider the following (quadratic) boolean function in $\B_m,$
$$f(x,x_m):\F_{2^{m-1}} \times  \F_2 \longrightarrow \F_2,$$
$$f(x,x_m)=\tr(\sum^{k}_{j=1}x^{2^j +1})+x_m \tr(x) \ \ (k=\frac{m}{2}).$$

For all $u \in \F_{2^{m-1}},$ let $f_u(x,x_m)=f(ux,x_m).$ It is known in \cite{BT} that
$$K_m=\{f_u(x,x_m) | u \in F_{2^{m-1}} \}, | K_m|=2^{m-1}$$
is a bent set in $\B_m,$ called Kerdock bent set since it is closely related to the Kerdock codes.
From $f_0(x,x_m)\equiv 0,$ we know that for each $u \in \F^{\ast}_{2^{m-1}},f_u(x,x_m)=f_u(x,x_m)+f_0(x,x_m)$
is also bent.

Let $L_m$ be the set of linear functions in $\F_{2^{m-1}} \times  \F_2.$ Namely,
$$L_m=\{ l_{a,a_m}(x,x_m)=\tr(ax)+a_m x_m | (a,a_m) \in \F_{2^{m-1}} \times  \F_2\},$$
$|L_m|=2^m.$

We get new larger subset of $\B_m,$
$$ S_m=\{f_u(x,x_m)+l_{a,a_m}(x,x_m) | (u,a,a_m) \in F^2_{2^{m-1}} \times \F_2\},$$
$|S_m|=2^{2m-1}.$ Then we get the following binary codes with length $n=2^m,$
$$\cC_m=\{c(u,a,a_m)=(f_u(x,x_m)+l_{a,a_m}(x,x_m))\}, $$
where $(u,a,a_m)$ run through $F^2_{2^{m-1}} \times \F_2.$

It is easy to see that
$$c(u+u^{\prime},a+a^{\prime},a_m+a_m^{\prime})=c(u,a,a_m)+c(u^{\prime},a^{\prime},a_m^{\prime}).$$
Therefore $\cC_m$ is a linear code and $\delta_r(\cC_m)=d_H(\cC_m).$ As construction B, we take
$\cC^{\ast}_m=\cC_m \backslash \{(0,0,0)\},$ the following computation shows that $\delta_r(\cC^{\ast}_m)$
is bigger than $d_H(\cC^{\ast}_m)$ if $r >1.$
\vspace{0.15cm}

\begin{theorem}\label{thm-caseC}
Let $m=2k \geq 4, n=2^m$ and
$\cC^{\ast}_m=\{c(u,a,a_m) \in \F^n_2 | (0,0,0)\neq (u,a,a_m)\in F_{2^{m-1}} \times F_{2^{m-1}} \times \F_2\}.$
Then the parameters of binary code $\cC^{\ast}_m$ is
$(n=2^m, K=|\cC^{\ast}_m|=2^{2m-1}-1,\delta_r),$ where
$\delta_r=\delta_r(\cC^{\ast}) \geq (r+1)2^{m-2}-(3r-1)2^{k-2}.$ Particularly,
$d_H(\cC^{\ast})=\delta_1(\cC^{\ast})=2^{m-1}-2^k$ and
$\delta_r(\cC^{\ast})-d_H(\cC^{\ast})=(r-1)2^{k-1}(2^{k-1}-3) >0$ when $k \geq 3$ and $r >1.$
\end{theorem}
\vspace{0.15cm}
\begin{proof}
For $c=c(u,a,a_m),c^{\prime}=c(u^{\prime},a^{\prime},a_m^{\prime})$ and
$(u,a,a_m)\neq (u^{\prime},a^{\prime},a_m^{\prime}),$
\begin{eqnarray*}
&&\delta_r(c,c^{\prime}) \\
&=& \frac{r+1}{4} 2^m\\
&-& \frac{r+1}{4} \left(\sum\limits_{(x,x_m)} (-1)^{f_{u+u^{\prime}}(x,x_m) +\tr((a+a^{\prime})x)+(a_m +a^{\prime}_m)x_m}       \right)\\
&&+ \frac{r-1}{4}\left(\sum\limits_{(x,x_m)} (-1)^{f_{u^{\prime}}(x,x_m) +\tr(a^{\prime}x)+a^{\prime}_m x_m} \right) \\
&&-\frac{r-1}{4}\left(\sum\limits_{(x,x_m)} (-1)^{f_{u}(x,x_m) +\tr(ax)+a_m x_m} \right)\\
&=&\frac{r+1}{4}(2^m -W_{f_{u+u^{\prime}}}(a+a^{\prime},a_m +a^{\prime}_m))\\
&&+\frac{r-1}{4}(W_{f_{u^{\prime}}}(a^{\prime},a^{\prime}_m)-W_{f_{u}}(a,a_m)).
\end{eqnarray*}

(1) If $u,u^{\prime} \in \F^{\ast}_{2^{m-1}}$ and $u \neq u^{\prime},$ then $f_u, f_{u^{\prime}}$
and $f_{u+u^{\prime}}$ are bent. We get
\begin{eqnarray*}
\delta_r(c,c^{\prime}) & \geq &\frac{r+1}{4}(2^m -2^k)-\frac{r-1}{4}2 \cdot 2^k \\
&=&(r+1)2^{m-2}-(3r-1)2^{k-2}=A.
\end{eqnarray*}

(2) If $u=u^{\prime} \in \F^{\ast}_{2^{m-1}},$
then $(a,a_m) \neq (a^{\prime},a^{\prime}_m),f_{u+u^{\prime}} \equiv 0$ and
\begin{eqnarray*}
&&\sum\limits_{(x,x_m)}(-1)^{f_{u+u^{\prime}}(x,x_m) +\tr((a+a^{\prime})x)+(a_m +a^{\prime}_m)x_m}\\
&=&\sum\limits_{x \in \F_{2^{m-1}}}(-1)^{\tr((a+a^{\prime})x)}\sum\limits_{x_m \in \F_{2}}(-1)^{(a_m+a_m^{\prime})x_m}\\
&=& 0.
\end{eqnarray*}
We get $\delta_r(c,c^{\prime}) \geq \frac{r+1}{4} 2^m-\frac{2r-2}{4} 2^k \geq  A.$

(3) If $u=0$ and $u^{\prime} \in \F^{\ast}_{2^{m-1}},$ then $(a,a_m) \neq (0,0)$ and
\begin{eqnarray*}
\delta_r(c,c^{\prime}) & = &\frac{r+1}{4}(2^m  -W_{f_{u^{\prime}}}(a+a^{\prime},a_m+a^{\prime}_m))\\
&&+\frac{r-1}{4}(W_{f_{u^{\prime}}}(a^{\prime},a^{\prime}_m)-0)\\
& \geq  & \frac{r+1}{4}(2^m  -2^k)-\frac{r-1}{4}2^k \\
&=& (r+1)2^{m-2}-r2^{k-1}>A.
\end{eqnarray*}
Similarly, if $u^{\prime}=0$ and $u \in \F^{\ast}_{2^{m-1}},$ we also have
$$\delta_r(c,c^{\prime}) \geq (r+1)2^{m-2}-r2^{k-1}>A.$$

(4) At last, $u=u^{\prime}=0.$
Then $(a,a_m)\neq (a^{\prime},a^{\prime}_m),(a,a_m)\neq(0,0),(a^{\prime},a^{\prime}_m)\neq (0,0).$
We get
$$\delta_r(c,c^{\prime})=\frac{r+1}{4}(2^m  -0)+\frac{r-1}{4}(0-0) =(r+1)2^{m-2} >A.$$

Therefore $\delta_r(\cC^k_m)\geq A=(r+1)2^{m-2}-(3r-1)2^{k-2}.$ Moreover, for
$(u,a,a_m) \neq (u^{\prime},a^{\prime},a^{\prime}_m),$
$$\delta_r(c,c^{\prime}) \geq A=(r+1)2^{k-1}(2^{k-1}-1)-(r-3)2^{k-2}>0,$$
which means that the codewords $c=c(u,a,a_m)$ and $c^{\prime}=c(u^{\prime},a^{\prime},a^{\prime}_m))$
are distinct. Therefore the size of $\cC^{\ast}_m$ is $K=|\F_{2^{m-1}}|^2 \cdot |\F_2|-1=2^{2m-1}-1.$
This completes the proof of Theorem \ref{thm-caseC}.
\end{proof}

\section{Conclusion}\label{sec-four}

Combinatorial neural codes are binary codes with asymmetric matched metric $\delta_r$
in stead of usual Hamming distance $d_H.$ The Hamming, Singleton and Plotkin bounds for
$\delta_r$ are presented in \cite{CR} in order to judge the goodness of
such codes using the theoretic neuroscience. In this paper we show that a binary code reaches the
Hamming, Singleton or Plotkin bound for $\delta_r$ if and only if the code reaches
the corresponding bound for $d_H$ and $r$ sufficiently closed to 1.

Since the parameters $(n,K,d_H)$ of binary codes reaching one of above three bounds
for $d_H$ are very limited, it is meaningful to present more combinatorial neural codes with
nice and flexible parameters $(n,K,\delta_r).$ In this paper we consider the methods to
construct binary codes from boolean functions given by Ding \cite{Ding2}, present three
constructions by bent functions. Many binary codes constructed
by other types of boolean functions with respect to Hamming distance $d_H$
have been done in past years (see survey paper \cite{LM}). It is hoped that the asymmetric
discrepacy $\delta_r$ can be determined or estimated for such codes
(omitted the zero codeword for linear code case by Theorem ). Moreover, other kinds of binary
codes can also be considered to have good and flexible parameters $(n,K,\delta_r).$
For example, we will show that binary irreducible cyclic codes are one of nice candidates
of combinatorial neural codes in sequential paper. The value of $\delta_r$ for such codes
can be computed by Gauss sums over the finite field extension over $\F_2.$



\end{document}